\documentclass[12pt]{article} 
\pdfoutput=1
\usepackage{amsmath,amsfonts,amsthm,amssymb,bbm,hyperref,eucal,verbatim,enumerate}
\usepackage[bbgreekl]{mathbbol}
\usepackage[OT2,OT1]{fontenc}
\usepackage{xcolor}

\newtheorem{thm}{Theorem}
\newtheorem{lemma}{Lemma}[section]

\newtheorem{conjecture}[lemma]{Conjecture}
\newtheorem{prop}[lemma]{Proposition}
\newtheorem{rmk}[lemma]{Remark}
\newtheorem{cor}[lemma]{Corollary}

\newtheorem*{bel*}{Belief}
\newtheorem*{assum*}{Assumptions}

\newtheoremstyle{named}{}{}{\itshape}{}{\bfseries}{.}{.5em}{\thmnote{#3}}
\theoremstyle{named}
\newtheorem*{thm1*}{Theorem}

\def\Lam{\mathbb{\Lambda}}
\def\od{\text{off-diag}}

\newcommand{\diam}{\operatorname{diam}}
\def\Sha{{\text{\fontencoding{OT2}\selectfont SH}}}
\def\Reg{{\bf Reg1}{\upshape)}--{\bf Reg2}{\upshape)}\,}
\newcommand{\be}{\begin{equation}}
\newcommand{\ee}{\end{equation}}

\newcommand{\Z}{\mathbb Z}

\newcommand{\dist}{\operatorname{dist}}
\newcommand{\mult}{\operatorname{mult}}
\newcommand{\gap}{\operatorname{gap}}
\newcommand{\tr}{\operatorname{tr}}

\newcounter{remnr1}
\def\res1{
    \addtocounter{remnr1}{1}
    \vspace{2mm}\noindent{\bf (\Alph{remnr1})} }

\numberwithin{equation}{section}

\begin{document}

\title{The trimmed Anderson model at strong
disorder: localisation and its breakup}
\author{Alexander Elgart\textsuperscript{1}, Sasha Sodin\textsuperscript{2}}
\footnotetext[1]{Department of Mathematics, Virginia Tech, Blacksburg, VA,
24061 USA.
 E-mail: aelgart@vt.edu.
Supported in part by NSF under grant  DMS-1210982.}
\footnotetext[2]{Department of Mathematics,
Princeton University,
Fine Hall, Washington Road, Princeton, NJ 08544-1000, USA \&
School of Mathematical Sciences, Tel Aviv University,
Ramat Aviv,
Tel Aviv 699780,
Israel.
E-mail: sashas1@post.tau.ac.il.
Supported in part by NSF under grant PHY-1305472.}
\maketitle

\begin{abstract}
We explore the properties of  discrete random Schr\"odinger operators in which the random
part of the potential is supported on a sub-lattice.
In particular, we provide new conditions on the
sub-lattice under which Anderson localisation
happens at strong disorder, and provide examples
in which it can be ruled out.
\end{abstract}

\section{Introduction}

In this paper we collect several observations
pertaining to the spectral properties
of random Schr\"odinger operators in the absence
of the so-called covering condition, which stipulates that the random potential is supported on the entire lattice. Let $\Lam$ be a lattice of bounded connectivity $\leq \kappa$,
and let
\begin{equation}\label{eq:defH} H(g) = -\Delta + V_0 + gV \end{equation}
be the operator acting on $\ell_2(\Lam)$ by
\begin{equation}
 [H(g)\psi](x) = \sum_{y \sim x} (\psi(x) - \psi(y)) + (V_0(x) + g V(x)) \psi(x)~, \quad x \in \Lam~.
 \end{equation}
Here we assume that $V_0: \Lam \to \mathbb{R}$ is a deterministic background potential,
$V: (\Omega \times) \Lam \to \mathbb{R}$ is a random potential that assumes independent identically distributed
entries with distribution $\mu$ on a sublattice $\Gamma \subset \Lam$, and $g \geq 0$ is a coupling constant.
Following \cite{EK}, we call (\ref{eq:defH}) a $\Gamma$-trimmed random Schr\"odinger
operator on $\Lam$. The usual Anderson model is recovered when $\Gamma = \Lam = \mathbb{Z}^d$.

\vspace{2mm}\noindent
Recall that the Anderson model exhibits localisation at strong disorder: for $g \gg 1$, the spectrum
is pure point and the eigenfunctions are exponentially localised. Two strategies of proof are available:
the first one, called multi-scale analysis, was devised by Fr\"ohlich and Spencer \cite{FS} and the second
one, the fractional moment method,--- by Aizenman and Molchanov \cite{AM}. Both have many
variants and ramifications, too numerous to be listed here, and surveyed, for example, by Figotin and
Pastur \cite[Chapter~15C]{FP}, Kirsch \cite{Kirsch}, and Stolz \cite{St}. We also mention  the work of Imbrie \cite{I} in which an iterative scheme to diagonalise the random operator is  suggested.

It is expected, on physical grounds  \cite{LTW}, that,
as the strength of the disorder decreases, the Anderson model undergoes a phase transition, and the absolutely continuous component of the spectrum emerges. From the mathematical physics perspective, the proof of such actuality remains one of the greatest challenges in the field.

\smallskip
The variant of the Anderson model which we consider in this work is characterized by two parameters: the strength $g$ of the disorder as in the standard Anderson model, and the sublattice $\Gamma$ of $\Z^d$ in which we insert the random potential. 

For $\Gamma = \Z^d$ we recover the usual Anderson
model with almost sure pure point spectrum for large $g$. We mainly consider the
case when $\Gamma$ is a periodic sublattice of $\Z^d$, and explore the dependence of the spectral
properties at strong disorder $g \gg 1$ on the geometry of $\Gamma$: when  $\Gamma$ is sufficiently dense (in
the  sense defined in Theorem \ref{thm:loc2} below), the behaviour is similar to that
of the usual Anderson model (Anderson localisation),
whereas for a sparser $\Gamma$ new phenomena appear, see discussion below.

Another direction (which we do not explore in depth here) is to choose $\Gamma$ at random,
according to the product probability measure (site
percolation). Then the case $g = \infty$ is known
as quantum percolation (see the paper of Veseli\'{c} \cite{Ves} for a survey of results). Finite $g > 0$ leads
to a model which combines the features of the Anderson 
model with those of quantum percolation. Thus one may expect
an interesting phase diagram as one varies both
the strength of disorder $g$ and the relative density\footnote{e.g. $ \limsup_{R\rightarrow\infty}|B(0,R)\cap\Gamma|/|B(0,R)|$, where $B(0,R)$ is a ball of radius $R$ as in \eqref{eq:sqbox}.}
of $\Gamma$; the results of the current paper
indicate how parts of this phase
diagram should look. In particular, our results suggest that the delocalisation
part of the phase diagram for such models may be more amenable
to analysis than in the usual Anderson model.

\medskip
 Our initial interest in the trimmed Anderson model was triggered by the following question. The known proofs of localisation make use of a priori estimates on the resolvent (Wegner-type bounds), and these
in turn require that the support of the potential is the entire lattice (covering condition). One may ask whether
localisation at strong disorder still holds when the covering condition is violated.

In the continuum setting, an affirmative answer to this question was established at the bottom of the spectrum using the unique continuation principle  (UCP), \cite{Kuc,RV} (Wegner bounds for such models were first established in \cite{CHK}). Although UCP is not applicable for the lattice Schr\"odinger operators, 
Rojas-Molina \cite{RM1} and Klein with the first author \cite{EK} developed Wegner estimates adjusted
for the trimmed Anderson model. These estimates allowed to prove localisation in the strong disorder regime, at the bottom of
the spectrum. In \cite{RM1}, the case of zero background $V_0 = 0$ was considered, whereas \cite{EK}
handled arbitrary bounded background potentials.

We make a further contribution in this direction, and prove (Theorems~\ref{thm:loc1} and \ref{thm:loc2})
localisation at strong disorder in several additional situations (not necessarily at the bottom of the spectrum).

\vspace{2mm}\noindent
Further, we explore the possible alternatives to localisation which may occur at strong disorder.

In certain situations, we prove (Theorem~\ref{thm:anom'}) that
sufficiently high moments associated with the
Green function diverge. Although this
phenomenon occurs only at a discrete set
of special energies, it implies (Lemma~\ref{l:anom})
the divergence of high moments associated
with the quantum dynamics, which is in turn
incompatible with strong forms of Anderson
localisation. This anomalous behaviour has previously
been rigorously observed only in one-dimensional
models, cf.\ Jitomirskaya,  Schulz-Baldes, and
Stolz \cite{JSS}.

One possibility is the
emergence of an absolutely continuous component
of the spectral measure about the special energies.
While we currently can not rigorously rule out this possibility, we find the following alternative (anomalous localisation) more
plausible: the spectral measure is pure point, however,
the localisation length of the eigenvectors diverges
at the special energies with a power-law singularity.
The quantum dynamics
picks up the contribution from all eigenvectors, therefore the position of the quantum
particle is a heavy-tailed random variable, and
its high moments diverge as the time grows.

A na{\"\i}ve classical analogue of this phenomenon
(ignoring the subtleties of quantum dynamics and also
the presence of multiple channels)
is the following: a particle moves along a circle of length $L$ with unit velocity, where $L$ is a heavy-tailed random variable. While this is a case of
localisation in any possible sense, sufficiently
high moments of the distance from the origin
at time $t$ diverge as $t$ grows to infinity.

\vspace{2mm}\noindent
Finally, in certain spectral regions the trimmed Anderson model at strong disorder can be coupled
to a weak disorder Anderson-type model, and this leads us to believe that in these regions the model
exhibits delocalisation in dimension $d \geq 3$.

\vspace{4mm}\noindent
Now let us state the results in more detail. Throughout the paper, we make the following three
\begin{assum*}\hfill
\begin{enumerate}
\item[{\bf Inv)}]  $\Lam$  is the 
$d$-dimensional lattice $\mathbb{Z}^d$;  the sublattice $\Gamma$ and  the background potential $V_0$ are invariant under a cofinite
subgroup $\mathcal{G} \subset\Z^d$.
\item[{\bf Reg1)}] The distribution $\mu$ is $\alpha$-regular for some $\alpha > 0$, meaning that, for any $\epsilon > 0$ and $t \in \mathbb{R}$,
$\mu[t - \epsilon, t + \epsilon] \leq C\epsilon^\alpha$.
 \item[{\bf Reg2)}] $\mu$ has a finite $q$-moment for some $q > 0$, meaning that
\[ M_q = \int |t|^q d\mu(t) < \infty~.\]
\end{enumerate}
\end{assum*}

The invariance assumption {\bf Inv)} is introduced mainly for convenience, and to inscribe the problem
into the familiar setting of ergodic (metrically transitive) random operators; it can
be mostly omitted or relaxed.
The regularity assumptions \Reg are essentially used in the arguments.

\subsection{Anderson localisation}

Denote by $G_z[H] = (H-z)^{-1}$ the resolvent of a self-adjoint operator $H$ acting on $\ell^2(\mathbb{Z}^d)$.
If the fractional moment bound
\begin{equation}\label{eq:fmbd}
\sup_{\epsilon > 0} \sup_{x \in \mathbb{Z}^d} \sum_{y \in \mathbb{Z}^d} \mathbb{E} |G_{\lambda + i\epsilon}[H](x, y)|^s e^{\eta \| x - y\|}< \infty
\end{equation}
holds for some $0 < s < 1$ and $\eta > 0$, we say that $H$ exhibits Anderson localisation at
$\lambda \in \mathbb{R}$.  Here $\| \cdot \|$ stands for the graph distance (i.e.\ the $\ell^1$ distance) on $\mathbb{Z}^d.$

The methods developed by Aizenman \cite{Aiz} (see further \cite{ASFH})
show that if (\ref{eq:fmbd})
holds for all values of $\lambda$ in an interval $I \subset \mathbb{R}$, then one has the following more physical dynamical
localisation for the spectral restriction $H|_I = \mathbf{P}_I[H] \, H \, \mathbf{P}_I[H]$ of the operator $H$ to $I$:
\begin{equation}\label{eq:dynloc}
\sup_{x \in \mathbb{Z}^d} \mathbb{E} \sup_{t \geq 0}  \sum_{y \in \mathbb{Z}^d} \left| e^{i t H|_I}(x, y) \right|^2 e^{\widetilde{\eta} \|x-y\|} < \infty~.\end{equation}
These methods
do not require major modification in the context of the current paper, therefore we focus on single-energy bounds
(\ref{eq:fmbd}).

Following the previous work \cite{EK}, we are interested in the following question: under which conditions
on $\Gamma$ and $\lambda$ does Anderson localisation hold at strong disorder, $g \gg 1$? As observed in \cite{EK},
the restriction
\[ H_\Gamma = P_{\Gamma^c} H P_{\Gamma^c}^*\]
of $H$ to the complement of $\Gamma$ plays an
important r\^ole (here $P_{\Gamma^c}: \ell_2(\mathbb{Z}^d) \to \ell_2(\Gamma^c)$ denotes coordinate projection).

\begin{thm}\label{thm:loc1} Let $H(g)$ be a $\Gamma$-trimmed random Schr\"odinger operator on $\mathbb{Z}^d$
satisfying the Assumptions. Suppose $\lambda \notin \sigma(H_\Gamma)$. Then
there exist $0 < s < 1$ and $g_0 > 0$ so that (\ref{eq:fmbd}) holds for all $g \geq g_0$.
\end{thm}
\begin{rmk}
It was shown in \cite{EK} that  $\inf_{E\in\sigma(H_\Gamma)}E>\inf_{E\in\sigma(H(g))}E$ almost surely, which implies that the statement above is non empty.
\end{rmk}
In section~\ref{s:loc1}, we prove the more general Proposition~\ref{prop:loc}, and deduce Theorem~\ref{thm:loc1}.
The proof is a relatively straightforward application of the fractional moment method of \cite{AM}.

\vspace{2mm}\noindent
The condition $\lambda \notin \sigma(H_\Gamma)$ is however not necessary for
Anderson localisation. To illustrate this, consider the
case when the complement of $\Gamma$ is a union
of finite connected components.  The following theorem
implies that, if the connected components are 
separated by a double layer of of sites in $\Gamma$ (``double insulation''), Anderson localisation holds
at all energies, including the eigenvalues of 
$H_\Gamma$.

\begin{thm}\label{thm:loc2} Let $H(g)$ be a $\Gamma$-trimmed random Schr\"odinger operator on $\mathbb{Z}^d$
satisfying the Assumptions. If $\Gamma^c$ is the union of finite connected components $B_j$ such that
$\dist(B_i, B_j) \geq 3$ for $i \neq j$, then there exist $0 < s < 1$ and $g_0>0$ such that (\ref{eq:fmbd}) holds for
all $g \geq g_0$ and all $\lambda \in \mathbb{R}$.
\end{thm}

The proof of Theorem~\ref{thm:loc2} appears in Section~\ref{s:loc2}; it is also based on the fractional
moment method, and makes use of a Wegner-type estimate which we prove in Section~\ref{s:weg1}. 

\vspace{2mm}\noindent
 The reason due to which double insulation forces localisation has to do with the fact that it rules out the existence of non-trivial formal solutions $\psi$ for $H$ which are supported on $\Gamma^c$. Therefore, in the case when the complement of $\Gamma$
is a union of finite connected components,
the following conjecture would be a generalisation 
of both Theorem~\ref{thm:loc1} and 
 Theorem~\ref{thm:loc2}.

\begin{conjecture} Suppose that the complement of $\Gamma$ is a union of finite connected
components, and that $\lambda \in \mathbb{R}$ is such that the eigenvalue equation
\begin{equation}\label{eq:eveq}
H(0) \psi = \lambda \psi
\end{equation}
has no non-trivial formal solution $\psi$ supported on $\Gamma^c$. Then (\ref{eq:fmbd}) holds for
sufficiently large $g$.
\end{conjecture}

\subsection{Anomalous localisation}\label{susec:aloc}
The situation is different when the eigenvalue equation (\ref{eq:eveq}) has a solution supported on $\Gamma^c$. Let us first consider the
case when all the connected components of $\Gamma^c$ are finite. We believe that, generically, in this situation
\begin{equation}\label{eq:pmoment}
\lim_{\epsilon \to +0} \sum_{y \in \mathbb{Z}^d} \epsilon^2 \mathbb{E}|G_{\lambda+i\epsilon}[H](x, y)|^2\|x-y\|^p\end{equation}
is infinite for sufficiently large $p > 0$. The following theorem confirms this belief under additional hypotheses.

\begin{thm}\label{thm:anom'}
Let $H(g)$ be a trimmed random Schr\"odinger operator satisfying the Assumptions, 
 with arbitrary $g > 0$, so that all the
connected components of $\mathbb{Z}^d \setminus \Gamma$ are finite. Fix $x \in \mathbb{Z}^d$, and suppose that there exist a sequence of connected finite subgraphs $B_n  \subset \mathbb{Z}^d$  and a pair of constants $C, c > 0$ such that
\begin{enumerate}
\item $B(x, R_n) \subset B_n \subset B(x, (R_n)^C)$,
where
 $R_n < R_{n+1} < (R_n)^C$ and 
 \be\label{eq:sqbox}B(x,R)=\{y\in\Z^d:\ \|y-x\|\le R\};\ee
\item\label{item:period'} there exists $y \in B_n$ such that $\|x-y\| \geq (R_n)^c$, and the spectral projection
$\mathbf{P}_{\{\lambda\}}[H_n(0)]$  onto the  eigenspace of the restriction $H_n(0)=P_{B_{n}} H(0) P_{B_{n}}^*$ corresponding to $\lambda$ satisfies
\begin{equation}\label{eq:lrbdproj}\left| \mathbf{P}_{\{\lambda\}}[H_n(0)] (x, y) \right|
\geq (R_n)^{-C}~; \end{equation}
\item ${\rm Range} \, \mathbf{P}_{\{\lambda\}}[H_n(0)]\subset \ell^2(\Gamma^c)$ ;
\item
$\min \left \{ |\lambda'  - \lambda| \, \mid \, \lambda' \in \sigma(H_n(0)) \setminus \{\lambda\} \right\} \geq(R_n)^{-C}$.
\end{enumerate}
Then $\text{(\ref{eq:pmoment})}=\infty$ for sufficiently large $p$.
\end{thm}

\begin{rmk} 
The second and third assumptions state that there exist
non-trivial  formal solutions of (\ref{eq:eveq})  
on large boxes, and that all these solutions
are supported on $\Gamma^{c}$. These conditions
imply in particular the existence of a non-trivial
formal solution on the entire lattice. The condition \eqref{eq:lrbdproj} implies that these  formal solutions exhibit  at most  power-law spatial growth, i.e.\ they are generalized eigenfunctions of $H(0)$, whereas the last assumption of the theorem asserts that the spectral gap between $\lambda$ and the rest of the spectrum decreases as a power of the size of the system, which is a generic condition for a periodic Schr\"odinger operator. Finally, the first assumption
is a mild regularity condition on the growth of 
the boxes.
\end{rmk}

The proof of Theorem~\ref{thm:anom'}
appears in Section~\ref{s:anom}.

\medskip\noindent
Let us present a couple of examples for the case of  zero background potential $V_0 =0$ 
in two dimensions $d = 2$. One can also construct examples in higher dimension
along the same lines.

The first example is 
\[ \Gamma_1(k,m) =  \left\{ x \in \mathbb{Z}^2 \, \mid \,   x_1  \in k \mathbb{Z}  \vee x_2 \in m\mathbb{Z}\right\}~,\]
where $k,m\geq 2$ is a pair of fixed natural numbers. In this case, any eigenfunction of the Dirichlet Laplacian
in the rectangular fundamental  cell $\{ x_1 \in \{1,\cdots,k-1\},\, x_2 \in \{1,\cdots,m-1\}\}$ can be extended (by reflection) to a periodic
eigenfunction of the Laplacian on $\mathbb{Z}^2$ which vanishes on $\Gamma_1(k,m)$.

The same is true for 
\[ \Gamma_2(k) = \left\{ x \in \mathbb{Z}^2 \, \mid \,  x_1 \in k \mathbb{Z} \vee x_2 - x_1 \in 2 \mathbb{Z} \right\}\]
when $k \geq 2$. In this example the fundamental cell is  a parallelogram.

\vspace{2mm}\noindent
Although Theorem~\ref{thm:anom'} proves the divergence of (\ref{eq:pmoment}) at a single energy
only, this is sufficient to imply that sufficiently
high moments
 \[ M_p(x, t) = \sum_{y \in \mathbb{Z}^d} \mathbb{E} |e^{itH}(x, y)|^2 \|x-y\|^p\]
associated with the unitary evolution (quantum
dynamics) $e^{itH}$
also diverge. Indeed, the following lemma holds
(see Section~\ref{s:anom} for the proof).
\begin{lemma}\label{l:anom}
For any $\lambda \in \mathbb{R}$, $\epsilon > 0$,
and $x \in \mathbb{Z}^d$,
\begin{equation}\label{eq:1ener} \int_0^\infty \epsilon e^{-\epsilon t} M_p(x, t) dt
\geq \sum_{y \in \mathbb{Z}^d} \mathbb{E} \epsilon^2 |G_{\lambda + i\epsilon}[H](x,y)|^2  \|x-y\|^p~.
\end{equation}
\end{lemma}

Thus, in the setting of Theorem~\ref{thm:anom'}, the moments $M_p(x, t)$ are unbounded (as $t \to \infty$) for sufficiently large $t$. We emphasise that
this behaviour is not necessarily a sign of delocalisation. If, as we assumed, all the components
of $\Gamma^c$ are finite, a solution $\psi$ of (\ref{eq:eveq}) supported on $\Gamma^c$ may exist only
for a discrete set of energies
$\lambda$. It is plausible that the operator $H(g)$ at strong disorder has pure point spectrum with exponentially
decaying eigenfunctions, and that the anomalous behaviour (\ref{eq:pmoment}) reflects the divergence
of the localisation length at the special energies. If this is the case, it is an instance of a phenomenon
sometimes referred to as anomalous
localisation, cf.\ the survey of Izrailev, Krokhin, and Makarov \cite{IKM}.

To establish anomalous localisation (as opposed to, say, the presence of
continuous spectrum in the vicinity of $\lambda$), one needs to complement Theorem~\ref{thm:anom'} with an upper bound
on (\ref{eq:pmoment}) for small $p>0$. We have not been able to accomplish this task. To the best of our
knowledge, anomalous localisation has to date only been proved in several one-dimensional models; we refer
in particular to the work of Jitomirskaya, Schulz-Baldes and Stolz \cite{JSS}.

\subsection{Delocalisation}\label{susec:deloc}
The third possibility that can occur in the invariant setting {\bf Inv)} is that $\Gamma^c$
is connected (or at least has an infinite connected component), and $\lambda$ lies
in a band of absolutely continuous spectrum of the periodic operator $H_\Gamma$.

\begin{conjecture}\label{c:deloc} Let $g \gg 1$, and let $I$ be an interval in the absolutely continuous spectrum of
$H_\Gamma$. If $d = 2$, $H(g)$ exhibits localisation (\ref{eq:fmbd}); when
$d \geq 3$, $H(g)$ has absolutely continuous spectrum on $I$.
\end{conjecture}

To support this conjecture, we  introduce in Section~\ref{s:hedge} a (rigorous) coupling between
random operators at strong and weak disorder. Similar ideas have been applied
in different context by Wang \cite{Wang}.

In Section~\ref{s:heur} we provide an heuristic argument
(making use of this coupling) in favour of  Conjecture~\ref{c:deloc}:  a trimmed Anderson operator at strong disorder is coupled to an Anderson-type operator at weak disorder in the same
dimension. If the Anderson-type operator exhibits localisation at $d = 2$ and delocalisation at $d = 3$
 (as one may believe based on  the conjectures for the usual Anderson model \cite{LTW}), the same properties are
inherited by the trimmed Anderson operator from
which we started.

\subsection{Other topics}

The following topics are also discussed in this paper.

First, the proofs of Theorems~\ref{thm:loc2} and \ref{thm:anom'} require somewhat
non-standard Wegner-type estimates, which we prove in Sections~\ref{s:weg1} and \ref{s:weg2}.

Second, as an additional application of the strong-to-weak disorder coupling of Section~\ref{s:hedge},
we provide a new proof of a theorem of Aizenman \cite{Aiz} (labelled here as
Theorem~\ref{thm:aiz}) on localisation at the spectral edges
at weak disorder.

\section{Preliminaries}

\subsection{Notation}\label{s:not}

Two sites $x, y\in \Lam$ are adjacent, $x \sim y$,
if they are connected by an edge.

If $B \subset \Lam$ is a subset of the lattice, the boundary $\partial B$ is the set of edges $(x, y)$ with $x \in B$ and $y \notin B$; denote
by $\partial_\text{in} B$ and $\partial_\text{out} B$ its projections onto the $x$- and $y$-coordinate, respectively. $P_B$ and $P_{B^c}$ denote the coordinate projections onto $B$ and its complement,
respectively.

Denote by  $\sigma(A)$ the spectrum of an operator $A$, and by $G_z[A] = (A-z)^{-1}$ the resolvent of $A$ (defined for $z \notin \sigma(A)$). If $A$ acts
on $\ell^2(\Lam)$, denote by
\[ G_z[A](x, y) = \langle \delta_x, (A-z)^{-1}\delta_y \rangle~, \quad x,y \in \Lam, \]
the matrix elements of the resolvent (the Green function).

If $A$ is self-adjoint and $J \subset \mathbb{R}$ is a
Borel set, we denote by $\mathbf{P}_J[A]$
the spectral projection on $J$. Sometimes we use the
notation
\[ \mathbf{Q}_J[A] = \mathbf{P}_{J^c}[A] = \mathbb{1} - \mathbf{P}_J[A]~.\]

Finally, we denote by $C$  a sufficiently
large positive constant, and by $c$ -- a sufficiently small positive constant; the values of $C$
and $c$ may change from line to line.

\subsection{Properties of the resolvent}

The following two formul{\ae} are especially useful for computing the Green function. The first one
is the
Schur--Banachiewicz formula: if $A$ is an invertible operator acting on $\ell^2(\Lam)$, $X\subset \Lam$,  then
\begin{equation}\label{eq:SB}
P_X A^{-1} P_X^* = \left( P_X A P_X^* - P_X A P_{X^c}^* \frac{1}{P_{X^c} A P_{X^c}^*} P_{X^c} A P_X^* \right)^{-1}~.
\end{equation}

 The second one is the resolvent identity, valid when $A$ is an operator of the form $A = -\Delta + U$ (the potential
 $U$ need not be real):
\begin{multline}\label{eq:resolv}
G_z[A](x, y) = \\
\begin{cases}
\displaystyle\sum_{X \ni u' \sim u \in X^c} G_z[A](x, u') G_z[A_X](u, y)~, & x\in X~, \, y \notin X \\
\displaystyle\sum_{X^c \ni u \sim u' \in X} G_z[A_X](x, u) G_z[A](u', y)~, & x\notin X~, \, y \in X \\
\displaystyle\sum_{\substack{X^c \ni u \sim u' \in X, \\ X \ni v' \sim v \in X^c}} G_z[A_X](x, u) G_z[A](u', v') G_z[A_X](v, y)~,
  &x,y \notin X~.
\end{cases}
\end{multline}

\vspace{2mm}\noindent
Next, we shall make use of the Combes--Thomas estimate \cite{CT}, which states that if $A = -\Delta + U$ is a
Schr\"odinger operator ($U$ is now real) on a lattice $\Lam$ of bounded connectivity, and $z \notin \sigma(H)$,
then $|G_z[A](x,y)|$ decays exponentially in $\dist(x, y)$:
\begin{equation}\label{eq:ct}
|G_z[A](x,y)| \leq C \exp(- c \dist(x, y)) \quad (z \notin \sigma(A))~,
\end{equation}
where the constants $C, c > 0$ depend only on the distance from $z$ to the spectrum of $A$ and on the connectivity
of the lattice. A version with a sharp dependence of $c$ on the distance from the spectrum was proved by
Barbaroux, Combes, and Hislop \cite{BCH}.

\subsection{Fractional moments: auxiliary estimates}

Here we cite two estimates which commonly appear in the applications of the
fractional moment
method, and go back to the original work of Aizenman and Molchanov \cite{AM}.

The first one is a decoupling inequality for rational functions. We cite it in the form of \cite[Proposition~3.1]{we},
which is  slightly more general than the original one of \cite[Appendix III]{AM} (where fractional-linear functions were considered).

\begin{lemma}\label{l:decoup} Let  $\mu$ be a probability measure on $\mathbb{R}$ satisfying the assumptions \Reg.
Let $a_1, \cdots, a_l, b_1, \cdots, b_m \in \mathbb{C}$, and let $s, r>0$ be such
that $rm < \alpha$ and $q \geq (sl+rm) \frac{\alpha}{\alpha - rm}$. Then
\[ \int \frac{\prod_{j=1}^l |v - a_j|^s}{\prod_{i=1}^m|v-b_i|^r} \, d\mu(v)
    \asymp \frac{\prod_{j=1}^l (1+|a_j|)^s}{\prod_{i=1}^m (1  +|b_i|)^r}~,\]
where the "$\asymp$" sign means that $\text{LHS} \leq C \, \text{RHS} \leq C' \,\text{LHS}$, and the numbers
$C,C'>0$ do not depend on the $a_j$ and $b_i$.
\end{lemma}

The following Wegner-type estimates are a restatement of those in \cite[Appendix II]{AM}:

\begin{lemma}\label{l:weg} Let $A$ be a random self-adjoint operator acting on $\ell^2(\Lam)$, and let
$x, y \in \Lam$.
\begin{enumerate}
\item If $A(x,x)$ is sampled from a measure $\mu$ obeying {\bf Reg1)} independently
of all the other entries of $A$, then $\mathbb{E} |G_z[A](x,x)|^s < C_s < \infty$ for any $s < \alpha$,
uniformly in $z \notin \mathbb{R}$.
\item If both $A(x,x)$ and $A(y,y)$ are sampled from a measure $\mu$ obeying {\bf Reg1)} independently
of each other and of the other entries of $A$, then also $\mathbb{E} |G_z[A](x, y)|^s < C_s < \infty$.
\end{enumerate}
\end{lemma}

\subsection{Fractional moments: decay of the resolvent}

It is convenient to express the decay of the off-diagonal
elements of the resolvent, and, more generally, of a kernel $A: X \times X \to \mathbb{C}$,
in terms of the following quantity $\chi$, which was introduced by Aizenman \cite{Aiz}, and which
quantifies the exponential decay of a kernel with respect to a metric. If $\rho$ is a metric on  $X\subset \Lam$,
set
\[ \chi_\rho(A) = \sup_{x \in X} \sum_{y \in X} e^{\rho(y, x)} |A(y, x)|~.\]
The expression $|A|^s$ will denote a point-wise power of the point-wise absolute value of the kernel
$A$, thus
\[ \chi_\rho(|A|^s) = \sup_{x \in X} \sum_{y \in X} e^{\rho(y, x)} |A(y, x)|^s~.\]
We denote
\[ \| \rho \| = \sup_{x \sim y} \rho(x, y)\]
 and assume (here and forth) that this quantity is finite.

The resolvent identity (\ref{eq:resolv}) implies the bounds
\begin{equation}\label{eq:res.chi*}
\chi_\rho(P_{X^c} G_z[A] P_{X^c}^*) \leq \kappa^2 e^{2\|\rho\|} \chi_\rho^2(G_z[A_X])
\chi_\rho(P_X G_z[A] P_X^*)
\end{equation}
and
\begin{equation}\label{eq:res.chi}
\chi_{\rho}(G_z[A]) \leq \kappa e^{\|\rho\|} \chi_{\rho}(G_z[A_X]) (1 + \kappa e^{\|\rho\|} \chi_\rho(G_z[A_X])
  \chi_\rho(P_XG_z[A]P_X^*)~.
\end{equation}

The next statement is a translation of \cite[Lemma~2.1]{AM} to the $\chi$-notation of \cite{Aiz}. We now set $X = \Lam$, and let $A^\od$ denote the off-diagonal part of a kernel $A$.

\begin{lemma}[Aizenman--Molchanov]\label{l:am}
Let $A$ be an operator acting on $\ell^2(\Lam)$, and let $V: \Lam \to \mathbb{R}$ be an in\-de\-pendent
identically distributed random potential satisfying the decoupling inequality
\begin{equation}\label{eq:decoup} \mathbb{E} \frac{|V(y) - a|^s}{|V(y) - b|^s}
\geq C_s^{-1} \mathbb{E} \frac{1}{|V(y)-b|^s}~, \quad a,b \in \mathbb{C}~, \, b \notin \mathbb{R}~.
\end{equation}
for some $0 < s < 1$ and $C_s > 0$. Let $\rho$ be a metric on $\Lam$ such that $\chi_\rho(|A^\od|^s)$ is finite.
Then, for
\[ g^s > C_s \chi_\rho(|A^\od|^s)~, \]
one has
\[ \chi_\rho(\mathbb{E} |G_z[A + gV]|^s) \leq \frac{C_s}{g^s - C_s \chi_\rho(|A|^s)}~.\]
\end{lemma}

\begin{proof}[Proof of Lemma~\ref{l:am}]
Let $x, y \in \Lam$. According to the definition of the resolvent, $G_z[A+gV](gV+A-z\mathbb{1}) = \mathbb{1}$, which can be written as
\[G_z[A+gV](x, y)(gV(y) + A(y, y)) = - \sum_{u \neq y} G_z[A+gV](x, u) A(u, y) + \delta(x-y)~. \]
Taking expectation of the $s$-moment and applying the inequality $|a+b|^s \leq |a|^s + |b|^s$, we obtain:
\begin{multline*}\mathbb{E} |G_z[A+gV](x, y)|^s |gV(y) + A(y, y)|^s
 \\\leq \sum_{u \neq y} \mathbb{E} |G_z[A+gV](x, u)|^s |A(u, y)|^s + \delta(x-y)~.
\end{multline*}
As a function of $V(y)$, the expression $G_z[A+gV](x, y)$ has the form
\[ G_z[A+gV](x, y) = a (V(y)-b)^{-1}~, \]
where $a,b$ may be random but do not depend on $V(y)$.
Therefore the decoupling estimate (\ref{eq:decoup}) yields the inequality
\[ \mathbb{E} |G_z[A+gV](x, y)|^s |gV(y)+A(y,y)|^s \geq C_s^{-1} g^s \mathbb{E} |G_z[A+gV](x, y)|^s~,\]
which implies
\[ \mathbb{E} |G_z[A+gV](x, y)|^s
  \leq \frac{C_s}{g^s} \left\{ \sum_{u \neq y} \mathbb{E} |G_z[A+gV](x, u)|^s |A(u, y)|^s + \delta(x-y) \right\}~. \]
Multiplying both sides by $e^{\rho(x, y)}$ and summing over $y \in \Lam$, we obtain:
\[ \chi_\rho(\mathbb{E} |G_z[A+gV]|^s)
  \leq \frac{C_s}{g^s} \Big\{ \chi_\rho(\mathbb{E} |G_z[A+gV]|^s) \, \chi_\rho(|A^\od|^s) + 1 \Big\}~.
\]
\end{proof}

\begin{rmk} Lemma~\ref{l:decoup} provides examples of distributions
satisfying (\ref{eq:decoup}), here and in Theorem~\ref{thm:aiz} below.
\end{rmk}

\section{Wegner estimates}\label{s:weg}

In this section we prove two Wegner-type estimates.

\subsection{First Wegner estimate}\label{s:weg1}

We start from a general property of discrete Schr\"odinger operators, cf.\ Bourgain and Klein \cite[\S 2.2]{BK}.

 \begin{lemma}\label{l:bk} Let $B \subset \mathbb{Z}^d$ be a finite box, and let $\partial_\text{in} B \subset B' \subset B$.
 If $\psi$ is an eigenvector of a random
Schr\"odinger operator $-\Delta|_B + U$ on $B$ with eigenvalue $\lambda$, and $x \in B$ is a site
with first coordinate
\[ x_1= \max_{y \in B} y_1- n~, \]
then
\[ |\psi(x)| \leq \sum_{y \in B'} |\psi(y)|
\sum_{S \in \mathfrak{S}_{xy}}   \prod_{u \in S} |U(u) + 2d  - \lambda|~,\]
where
\[ \mathfrak{S}_{xy} \subset \left\{ u \in B \, \mid \, u_1 > x_1 \right\}~, \quad
\sum_{y} \# \mathfrak{S}_{xy} \leq (2d)^n~, \]
every $S \in \mathfrak{S}_{xy}$ is of cardinality $\# S \leq n$, and $\# S \cap B' \leq 1$.
\end{lemma}

\begin{proof} With the convention that an empty product is equal to one, the estimate holds for
$n=0$ and, more generally, for $x \in B'$. If $x \notin B'$, we proceed
by induction. The  eigenvalue equation at $x' = x + e_1$ yields
\[
|\psi(x)|
\leq |2d + U(x') - \lambda| |\psi(x')| + \sum_{w \sim x', w \neq x} |\psi(w)|~, \]
whence the claim follows with
\[ \mathfrak{S}_{xy} = \left\{ S \cup \{ x' \} \, \mid \, S \in \mathfrak{S}_{x'y} \right\}
\cup \bigcup_{w \sim x', w \neq x} \mathfrak{S}_{wy}~. \]
\end{proof}

In the context of trimmed random Schr\"odinger operators, Lemma~\ref{l:bk}
implies:
\begin{cor}\label{cor:weg1}
Let $B \subset \mathbb{Z}^d$ be a finite box such that $\partial_\text{in}B \subset \Gamma$. Then, for sufficiently
small $s$, the restriction $H(g)|_B = P_B H(g) P_B^*$ of an operator $H(g)$ satisfying the assumptions
\Reg
admits the estimate
\[ \mathbb{E} \|G_z[H(g)|_B]\|^s \leq C(B,g)~, \quad z \notin \mathbb{R}~. \]
\end{cor}
\begin{proof}
Let $\{ \psi \}_{\psi \in \Psi}$ be the eigenfunctions of $H(g)|_B$. Then, for every $x \in B$,
\[ 1 = \sum_{\psi \in \Psi} |\psi(x)|^2 \leq C_B' \sum_{y \in \partial_{\text{in}}B} |\psi(y)|^2 \sum_{S \in \mathfrak{S}_{xy}}   \prod_{u \in S} |U(u) + 2d  - \lambda|^2~,\]
where $U = V_0 + gV$,
 hence for $\Re w = \Re z = \lambda$
\begin{multline*}
\Im \tr G_w[H(g)|_B] \leq \\
C_B' \sum_{y \in \partial_{\text{in}}B} |\psi(y)|^2 \sum_{S \in \mathfrak{S}_{xy}}   \prod_{u \in S} |U(u) + 2d  - \lambda|^2 \Im \tr P_{\partial_\text{in}B} G_w[H(g)|_B] P_{\partial_\text{in}B}^*~.
\end{multline*}
The proof is concluded by taking the $s$-th moment with sufficiently small $s>0$
(note that it is sufficient to establish the bound for $\lambda$ in a compact
interval depending on $B$ and $g$).
 \end{proof}

Corollary~\ref{cor:weg1} yields the bound
\begin{equation}\label{eq:weg1}
\mathbb{E} |G_z[H(g)|_B](x, y)|^s \leq C(B, g)~.
\end{equation}
The constant $C(B, g)$ grows exponentially in the diameter of $B$, and as $1 + g^s$ in $g$. Note that, for $x, y \in \Gamma \cap B$, Lemma~\ref{l:weg} yields the better estimate
\begin{equation}\label{eq:weg0}
\mathbb{E} \left|G_z[H(g)|_B](x, y) \right|^s \leq C g^{-s}
\end{equation}
(the factor $g^{-s}$ is due to the normalisation which is different from that of
Lemma~\ref{l:weg}).

\subsection{Second Wegner estimate}\label{s:weg2}

The next deterministic lemma holds for any $H(g) = H(0)+gV$ with $V$ supported on $\Gamma$.

Let $B \subset \mathbb{Z}^d$ be a finite box, and let $H(g)|_B = P_B H(g) P_B^*$ be the restriction
of $H(g)$ to $B$. Denote by $\mult_\lambda$ the multiplicity of $\lambda$ in the spectrum of $H(0)|_B$,
and by $\gap_\lambda$ the distance from $\lambda$ to $\sigma(H(0)|_B) \setminus \{\lambda\}$.

\begin{lemma}\label{l:weg2}Suppose all the eigenvectors of $H(0)|_B$ corresponding to $\lambda$ are supported
on $B \setminus \Gamma$.  If $\phi$ is a normalised eigenvector of $H(g)|_B$ corresponding
to $\lambda'$, where $|\lambda - \lambda'| \leq \gap_\lambda / 3$ such that
$\phi \perp \operatorname{Ker} (H(0)|_B - \lambda)$,
then
\begin{equation}\label{eq:contrB'}
\|\phi|_{B\cap\Gamma}\| \geq \frac{\gap_\lambda}{3g\|V|_B\|_\infty}~.
\end{equation}
\end{lemma}

\begin{proof}[Proof of Lemma~\ref{l:weg2}] If (\ref{eq:contrB'}) fails,
\[\begin{split}
\| (H(0)|_B - \lambda)\phi \|
&\leq  \|(\lambda - \lambda')\phi\| + \|(H(0)|_B-H(g)|_B)\phi\| \\
&\leq \gap_\lambda / 3 + \gap_\lambda/3 < \gap_\lambda~,
\end{split}\]
in contradiction with the assumption. Thus (\ref{eq:contrB'}) is proved.
\end{proof}

\begin{lemma}\label{l:weg2'} Assume that $H(g)$ satisfies \Reg,
and that all  the eigenvectors of $H(0)|_B$ corresponding to
$\lambda$ are supported on $B \setminus \Gamma$. Then, for all $\epsilon \leq \gap_\lambda / 3$ and sufficiently small $s>0$,
\[ \mathbb{P} \left\{ \text{$H(g)|_B$ has $>\mult_\lambda$ eigenvalues in $(\lambda-\epsilon, \lambda+\epsilon)$} \right\} \leq \frac{C \epsilon^s g^s}{\gap_\lambda^{2s}} \left(\# B\cap\Gamma\right)^2~. \]
\end{lemma}

\begin{proof}
Suppressing the dependence of
the spectral projectors on the operator $H(g)|_B$,
we have:
\begin{align*}\tr \mathbf{P}_{[\lambda-\epsilon,\lambda+\epsilon]}&=\tr \mathbf{P}_{\{\lambda\}}\mathbf{P}_{[\lambda-\epsilon,\lambda+\epsilon]}+\tr \mathbf{Q}_{\{\lambda\}}\mathbf{P}_{[\lambda-\epsilon,\lambda+\epsilon]}  \\
&\le \tr \mathbf{P}_{\{\lambda\}}+2\epsilon\Im \tr \mathbf{Q}_{\{\lambda\}} \mathbf{P}_{[\lambda-\epsilon,\lambda+\epsilon]} G_{\lambda + i\epsilon}[H(g)|_B]\mathbf{P}_{[\lambda-\epsilon,\lambda+\epsilon]}\mathbf{Q}_{\{\lambda\}}\\
&\le  \mult_\lambda+\frac{18 \epsilon g^2 \|V|_B\|^2_\infty}{\gap_\lambda^2}
\Im \tr P_{\Gamma} G_{\lambda+i\epsilon}[H(g)|_B] P_{\Gamma}^*~,
\end{align*}
where we used Lemma~\ref{l:weg2} and the expansion
\[ \Im \tr G_{\lambda + i\epsilon}[H(g)|_B] =
\sum_j \sum_{x \in B} \frac{\epsilon |\psi_j(x)|^2}{(\lambda_j - \lambda)^2 + \epsilon^2} \]
over eigenvectors $H(g)|_B\psi_j = \lambda_j \psi_j$.

\vspace{2mm}\noindent
Let
\begin{equation}\label{eq:defN}
N=\tr \mathbf{P}_{[\lambda-\epsilon,\lambda+\epsilon]}-\mult_\lambda~,
\end{equation}
then
\begin{align} N
&\leq \frac{18 \epsilon g^2 \|V|_B\|^2_\infty}{\gap_\lambda^2}
\Im \tr P_{\Gamma} G_{\lambda+i\epsilon}[H(g)|_B] P_{\Gamma}^*\\
&\leq \frac{18 \epsilon g^2 \|V|_B\|_2^2}{\gap_\lambda^2}
\sum_{x \in B \cap \Gamma} |G_{\lambda+i\epsilon}[H(g)|_B](x, x)|~.
\end{align}
Therefore
\[ N^s \leq  \frac{18^s \epsilon^s g^{2s}}{\gap_\lambda^{2s}} \sum_{x,y \in B \cap \Gamma} |V(y)|^{2s} \left|G_{\lambda+i\epsilon}[H(g)|_B](x, x)\right|^s~. \]
Integrating over the distribution of $V(x)$ and using the Cauchy--Schwarz inequality
to decouple the potential from the Green function and then the Wegner-type estimate
in the first item of Lemma~\ref{l:am} in conjunction with Lemma~\ref{l:decoup} to bound the latter,
we conclude that
\[ \mathbb{P} \left\{ N \geq 1 \right\} \leq \mathbb{E} N^s \leq \frac{C \epsilon^s g^s}{\gap_\lambda^{2s}} \left( \# \, \Gamma \cap B\right)^2~. \]
\end{proof}
We furthermore obtain:
\begin{cor}\label{cor:weg2}
Assume that $H(g)$ satisfies \Reg, and that all the  eigenvectors of $H(0)|_B$ corresponding to
$\lambda$ are supported on $B \setminus \Gamma$. Then, for sufficiently small $s_0>0$ and any $0 < s < s_0$,
\[ \mathbb{E} \|\mathbf{Q}_{\{\lambda\}} G_{\lambda+i\epsilon}[H(g)|_B] \mathbf{Q}_{\{\lambda\}}  \|^s \leq \frac{C g^{s_0} (\# \Gamma\cap B)^2}{\gap_\lambda^{2s_0}} + \frac{C}{\gap_\lambda^{s_0}}~. \]
\end{cor}

\begin{proof}
We have:
\[ \mathbb{E} \|\mathbf{Q}_{\{\lambda\}}  G_{\lambda+i\epsilon}[H(g)|_B] \mathbf{Q}_{\{\lambda\}}  \|^s
= \int_0^\infty \mathbb{P} \left\{ \|\mathbf{Q}_{\{\lambda\}}  G_{\lambda+i\epsilon}[H(g)|_B] \mathbf{Q}_{\{\lambda\}}  \|^s  \geq t \right\} dt~.\]
For $t \leq (3/\gap_\lambda)^{s_0}$ we bound the integrand by $1$, and for larger $t$ we use Lemma~\ref{l:weg2'} (with $s_0$ in place of $s$).
\end{proof}

\section{Localisation}\label{s:loc}

\subsection{Outside the spectrum of $H_\Gamma$}\label{s:loc1}

In this section we prove Theorem~\ref{thm:loc1}. It will be convenient to drop the assumption
{\bf Inv)}, and to work on a general lattice $\Lam$ which we only assume to have bounded
connectivity $\leq \kappa$.

 If $X \subset \Lam$, let
$T_X: \ell^2(X^c) \to \ell^2(X)$ be the adjacency operator,
\[ T_X(x, y) = \begin{cases} 1, & x \sim y \\ 0, &\text{otherwise} \end{cases}~.\]
The condition for localisation is expressed in terms of the kernel
\begin{equation}\label{eq:defK}
K= (P_\Gamma \Delta P_\Gamma^* + T_\Gamma G_z[H_\Gamma] T_\Gamma^*)^\od~.
\end{equation}

\begin{prop}\label{prop:loc}
Let $H(g)$ be a $\Gamma$-trimmed random Schr\"odinger operator satisfying \Reg on a lattice $\Lam$ of
connectivity $\leq \kappa$. For any $0 < s < \alpha(1 + 2\alpha q^{-1})^{-1}$ there exists $C_s>0$ that may depend on $s$ and the constants in \Reg, such that
the following holds: if
\[ g^s > C_s \chi_\rho(|K|^s)~,\]  then
\begin{multline*}\chi_{\rho}(\mathbb{E} |G_z[H(g)]|^s)
  \\\leq \frac{C_s \kappa e^{\|\rho\|}}{g^s - C_s \chi_{\rho}(|K|^s)}
   \chi_{\rho}(|G_z[H_\Gamma]|^s) \left\{ 1 + \kappa e^{\|\rho\|} \chi_{\rho}(|G_z[H_\Gamma]|^s) \right\}~.
\end{multline*}
\end{prop}

Let us show that Proposition~\ref{prop:loc} implies Theorem~\ref{thm:loc1}.

\begin{proof}[Proof of Theorem~\ref{thm:loc1}]
Suppose  $H(g)$ satisfies the Assumptions.
According to the Combes--Thomas estimate (\ref{eq:ct}), $G_z[H_\Gamma]$ decays exponentially for
 $\lambda \notin \sigma(H_\Gamma)$, therefore $\chi_\rho[\left|G_z[H_\Gamma]\right|^s]$ is finite when
 $\rho$ is a small multiple of the graph metric on $\mathbb{Z}^d$, and hence so is
 \[ \chi_{\rho}(|K|^s) \leq \kappa e^{\|\rho\|} + \kappa^2 e^{2\|\rho\|} \chi_{\rho}(|G_z[H_\Gamma]|^s)~.\]
 According to Proposition~\ref{prop:loc}, $\chi_{\rho}(\mathbb{E} |G_z[H(g)]|^s)$ is finite for sufficiently
 large $g$, therefore (\ref{eq:fmbd}) holds.
\end{proof}

Now we prove Proposition~\ref{prop:loc}.
\vspace{2mm}\noindent
\begin{proof}[Proof of Proposition~\ref{prop:loc}]
Using the Schur--Banachiewicz formula \eqref{eq:SB} we get
\[ P_\Gamma G_z[H] P_\Gamma^* = G_z[gV|_\Gamma - D - K]~, \]
where $D+K$ is the decomposition of
\[ P_\Gamma \Delta P_\Gamma^* -  V_0|_\Gamma + T_\Gamma G_z[H_\Gamma] T_\Gamma^* \]
into diagonal and off-diagonal parts (this notation is consistent with (\ref{eq:defK})). According to
the Aizenman--Molchanov estimate (Lemma~\ref{l:am}), for
\[ g^s > C_s \chi_\rho(|K^\od|^s) \]
we have:
\[ \chi_\rho(P_\Gamma G_z[H] P_\Gamma^*) \leq \frac{C_s}{g^s - C_s \chi_\rho(|K^\od|^s)} \]
(the assumption (\ref{eq:decoup}) is satisfied according to Lemma~\ref{l:decoup}.) The proposition now
follows from the corollary (\ref{eq:res.chi}) of the resolvent identity.
\end{proof}

\subsection{Double insulation} \label{s:loc2}

\begin{proof}[Proof of Theorem~\ref{thm:loc2}]
Let us partition the lattice $\mathbb{Z}^d$ into disjoint pieces $B$:
\[ \mathbb{Z}^d = \biguplus_{B \in \mathfrak{B}} B~, \]
such that $\diam B \leq \mathrm{const}$ and $\partial_{\text{in}} B \subset \Gamma$ for
every $B \in \mathfrak{B}$. We do not
make any additional assumptions on the shape
of $B \in \mathfrak{B}$.

Let $x, y \in \mathbb{Z}^d$. Applying the resolvent identity (\ref{eq:resolv}), we can represent $G_z[H](x, y)$
as a sum of terms of the form
\[ G_z[H|_{B_1}](x, u_1) G_z[H|_{B_2}](u_1', u_2) G_z[H|_{B_3}](u_2', u_3) \cdots G_z[H|_{B_n}](u_{n-1}', y)~, \]
where $B_1 \ni \{x, u_1\}$, $B_2 \ni \{u_1', u_2\}$, \dots, $B_n \ni \{u_{n-1}', y\}$ are distinct boxes,
and $u_j$ is adjacent to $u_j'$. In particular, $u_j, u_{j-1}' \in \partial_{\text{in}}B_j \subset \Gamma$.

Taking fractional moments, we obtain:
\begin{multline*}\mathbb{E} |G_z[H](x, y)|^s \leq \\
\sum \mathbb{E} |G_z[H|_{B_1}](x, u_1)|^s \,\,\, \prod_{j=2}^{n-1} \mathbb{E} | G_z[H|_{B_j}](u_{j-1}', u_j)|^s
 \,\,\, \mathbb{E}|G_z[H|_{B_n}](u_{n-1}', y)|^s~.
\end{multline*}
For small $s>0$, we bound the first and last term by $C(B_j, g)$ using (\ref{eq:weg1}), and all the other terms
by $\operatorname{const} g^{-s}$ using (\ref{eq:weg0}). For large $g \geq g_0$, the resulting expansion
converges, and is exponentially decaying in $\|x-y\|$.
\end{proof}

\section{Anomalous localisation}\label{s:anom}

\begin{proof}[Proof of Theorem~\ref{thm:anom'}]
Set $\mathbf{P}_{\{\lambda \}}= \mathbf{P}_{\{\lambda\}}[H_n(0)]$, and $\mathbf{Q}_{\{\lambda \}}= \mathbf{Q}_{\{\lambda\}}[H_n(0)]$,
and let
\[ G_n^\pm = G_{\lambda \pm i\epsilon}[H_n(g)]~, \quad  G_n^\Re = \frac{G_n^+ + G_n^-}{2}~, \quad
G_n^\Im = \frac{G_n^+ - G_n^-}{2i}~. \]

\smallskip\noindent
Suppose in contrapositive that the assertion of the theorem is false. Then, for any $p > 0$, one can find a sequence $\epsilon_j \searrow 0$ such that the inequality
\[ \mathbb{E}|G_{\lambda+i\epsilon}[H (g)](x, v)|^s  \le M_p \epsilon^{-s} \|x-v\|^{-sp/2}\]
holds for all $\epsilon = \epsilon_j$ and all $x,v \in \mathbb{Z}^d$. For any $A >0$ we can find  $n = n_j$  so that $(R_n)^{-CA}\leq \epsilon \leq (R_n)^{-A}$. We shall choose the value of $A$ in the sequel.

For a fixed $x\in\Z^d$, consider a site $y\in \Z^d$ that  satisfies  $\|x-y\| \geq (R_n)^c$. Using the first resolvent identity, we can estimate
\begin{equation}\label{eq:bd1}\begin{split}
&\mathbb{E}|G_n^\Im(x,y)|^s\\
&\le \mathbb{E}|G_{\lambda+i\epsilon}[H (g)](x, y)|^s +\sum_{\langle v,u\rangle\in\partial B_n}\mathbb{E}| G_{\lambda+i\epsilon}[H(g)]( x,v)|^s| G_n^+(u,y)|^s\\ 
&\le \mathbb{E}|G_{\lambda+i\epsilon}[H(g)]( x,y)|^s+2d \epsilon^{-s}\sum_{ v\in\partial^{\text{in}} B_n}\mathbb{E}| G_{\lambda+i\epsilon}[H(g)]( x,v)|^s\\ 
&\le  M_p\left\{ \epsilon^{-s} (R_n)^{-csp/2} + 2d \epsilon^{-2s} (R_n)^{-sp/2} \right\}~.
\end{split}\end{equation}

\smallskip\noindent
On the other hand, according to Assumption~2 of the theorem,  there exists $y$ which satisfies $\|x-y\| \geq {(R_n)}^c$ and
\begin{equation}\begin{split}\label{eq:grow}
\mathbb{E}|G_n^\Im(x,y)|^s
&= \mathbb{E} |\mathbf{P}_{\{\lambda \}}G_n^\Im\mathbf{P}_{\{\lambda \}}(x,y) +
\mathbf{Q}_{\{\lambda \}}G_n^\Im\mathbf{Q}_{\{\lambda \}}(x,y)  |^s\\
&\geq |\mathbf{P}_{\{\lambda \}}G_n^\Im\mathbf{P}_{\{\lambda \}}(x,y)|^s - 
\mathbb{E} |\mathbf{Q}_{\{\lambda \}}G_n^\Im\mathbf{Q}_{\{\lambda \}}(x,y)  |^s~.\end{split}\end{equation}
According to Assumption~3 of the theorem, 
\[ f(H_n(g)) \mathbf{P}_{\{\lambda\}}[H_n(0)]=f(\lambda) \mathbf{P}_{\{\lambda\}}[H_n(0)]\] for any function $f$, therefore the assumption (\ref{eq:lrbdproj}) allows to  bound
the first term of (\ref{eq:grow}) from below by 
$\epsilon^{-s} R_n^{-C}$. If $A$ is chosen to
be sufficiently large, Corollary~\ref{cor:weg2} implies
that the second term is bounded from above
by one half of this quantity. Therefore
\[ \mathbb{E}|G_n^\Im(x,y)|^s \ge \epsilon^{-s} (R_n)^{-C} / 2~.\]
For sufficiently large $p$, this lower bound is
in contradiction with the upper bound (\ref{eq:bd1}).
\end{proof}

We conclude this section with the
\begin{proof}[Proof of Lemma~\ref{l:anom}]

We start from the identity
\[ \int_0^\infty e^{it(H-\lambda + i\epsilon)} dt = i G_{\lambda - i\epsilon}[H]~, \]
which implies:
\[ | G_{\lambda - i\epsilon}[H](x, y) | \leq \int_0^\infty |e^{itH} (x, y)| e^{-\epsilon t} dt~. \]
Taking the expectation and applying the Cauchy--Schwarz inequality, we obtain:
\[ \epsilon^2 \mathbb{E} | G_{\lambda - i\epsilon}[H](x, y) |^2
  \leq \int_0^\infty \epsilon e^{-\epsilon t} \mathbb{E} |e^{itH}(x, y)|^2 dt~.\]
This proves (\ref{eq:1ener}), since the sign of $\epsilon$ does not affect the absolute value.
\end{proof}

\section{Strong-to-weak disorder coupling}\label{s:s2w}

In this section we construct a coupling between a random operator
at strong disorder and another one at weak disorder. A similar coupling
appears in the work of Wang \cite{Wang}, who used it to construct examples
of long-range operators with exponentially decaying Green function.

\subsection{The hedgehog lattice}\label{s:hedge}

Let $\Lam$ be a lattice. Construct the hedgehog lattice $\Lam^\Sha = \Lam \times \{0, 1\}$
with bonds defined by
\[ (x, i) \sim (y, j) \iff \begin{cases}
\text{either $x = y$ and $i = 1-j$} \\
\text{or $x \sim y$ and $i=j=0$} \end{cases}~.\]
Given an operator $H(0)$ on $\Lam$ and a potential $U: \Lam \to \mathbb{C}$, consider the operator
$H^\Sha$ on $\ell^2(\Lam^\Sha)$, defined by
\begin{equation}\label{eq:defSha} H^\Sha \left( \begin{matrix} \psi_1 \\ \psi_0 \end{matrix} \right)
  = \left( \begin{matrix} U & - \mathbb{1} \\ - \mathbbm{1} &  H(0) \end{matrix} \right)
  \left( \begin{matrix} \psi_1 \\ \psi_0 \end{matrix} \right)~.
  \end{equation}
Observe that $H^\Sha$ (or rather, $H^\Sha + \mathbb{1}$) is a $(\Lam \times \{1\})$-trimmed random
Schr\"odinger operator on $\ell^2(\Lam^\Sha)$, if the values of $U$ are independent and real.

The Schur--Banachiewicz formula (\ref{eq:SB}) relates the resolvent of $H^\Sha$ to the resolvents of
$H = H(0) + U_z^\#$,  $U_z^\# = (z-U)^{-1}$, on $\Lam \times \{0\}$, and $H' = - G_z[H(0)] + U$ on $\Lam \times \{1\}$.
Namely, set $P_0 = P_{\Lam \times \{0\}}$, $P_1 = P_{\Lam \times \{1\}}$. Then
\begin{equation}\label{eq:s2w}\begin{split}
P_0 G_z[H^\Sha] P_0^* &= G_z[H(0) + U_z^\#]~,   \\
P_1 G_z[H^\Sha] P_1^* &= G_z[-G_z[H(0)] + U]~.
  \end{split}\end{equation}

\vspace{2mm} The first application of the relations (\ref{eq:s2w}) is a new derivation
of a theorem of Aizenman \cite{Aiz} which provides a sufficient condition for localisation
at weak disorder near the spectral edges.

Let $V:\Lam \to \mathbb{R}$ be a random potential, and consider the operator
$H(g) = H(0) + gV$ on $\ell^2(\Lam)$.

\begin{thm}[Aizenman]\label{thm:aiz}
Fix $0 < s < 1$, and suppose the random potential $V:\Lam \to \mathbb{R}$ satisfies the decoupling
property
\begin{equation}\label{eq:decoup1} \mathbb{E} \frac{|V(y) - a|^s}{|V(y) - b|^s}
\geq C_s^{-1} \mathbb{E} \frac{1}{|V(y)-b|^s}~, \quad a,b \in \mathbb{C}~, \, b \notin \mathbb{R}~.
\end{equation}
 Let $\rho$ be a metric on $\Lam$ so that
\[ \chi = \limsup_{\epsilon \to +0} \chi_\rho(|G_{\lambda + i\epsilon}[H(0)]|^s)  < \infty~.\]
Then, for $g^{-s} > C_\mu \chi$,
one has
\[ \limsup_{\epsilon \to +0} \chi_\rho(\mathbb{E} |G_{\lambda + i\epsilon}[H(0)+gV]|^s)
\leq\frac{C_\mu \kappa^2 e^{2\|\rho\|} \chi^2}{g^{-s} - C_\mu \chi}~. \]
\end{thm}

\begin{proof}
Let $U = \lambda + i\epsilon - \frac{1}{gV}$, and construct the operator $H^\Sha$
associated with $U$ as in (\ref{eq:defSha}). We have: $U^\#_{\lambda + i\epsilon} = gV$, therefore
the second half of (\ref{eq:s2w}) yields:
\[P_1 G_{\lambda+i\epsilon}[H^\Sha_\epsilon] P_1^* = G_{\lambda+i\epsilon}[-G_{\lambda+i\epsilon}[H(0)] + U]~.\]
By Lemma~\ref{l:am}, if
\[ g^{-s} > C_\mu \chi_\rho(|G_{\lambda + i\epsilon}[H(0)]|^s)~, \]
we have:
\[ \chi_\rho(\mathbb{E} |P_1 G_{\lambda+i\epsilon}[H^\Sha_\epsilon] P_1^*|^s)
\leq \frac{C_\mu}{g^{-s} - C_\mu \chi_\rho(|G_{\lambda + i\epsilon}[H(0)]|^s)}~.\]
According to the corollary (\ref{eq:res.chi*}) of the resolvent identity,
\[\begin{split}
&\chi_\rho\Big(\mathbb{E} |P_0 G_{\lambda+i\epsilon}[H^\Sha_\epsilon] P_0^*|^s\Big) \\
  &\quad\leq \kappa^2 e^{2\|\rho\|} \chi_\rho^2(G_{\lambda+i\epsilon}[H(0)])
    \chi_\rho(\mathbb{E} |P_1 G_{\lambda+i\epsilon}[H^\Sha_\epsilon] P_1^*|^s) \\
  &\quad\leq\frac{C_\mu \kappa^2 e^{2\|\rho\|}  \chi_\rho^2(G_{\lambda+i\epsilon}[H(0)])  }
                 {g^{-s} - C_\mu \chi_\rho(|G_{\lambda + i\epsilon}[H(0)]|^s)}~.
\end{split}\]
Applying the first half of (\ref{eq:s2w}) and taking the upper limit as $\epsilon \to + 0$, we conclude
the proof.
\end{proof}

\subsection{Trimmed random Schr\"odinger operators}\label{s:heur}

In this short section, we use the strong-to-weak disorder coupling to provide non-rigorous support
for Conjecture~\ref{c:deloc}.

We apply the strong-to-weak disorder relations (\ref{eq:s2w}) in the direction opposite to that of
Section~\ref{s:hedge}. First consider
the hedgehog lattice ($\Lam \times \{1\}$)-trimmed operator $H^\Sha + \mathbb{1}$ corresponding to $U = gV$,
$g \gg 1$. The first part of (\ref{eq:s2w}) relates the resolvent of $H^\Sha$ to the resolvent of the operator
$H(0) + (gV)_z^\#$.

Now consider the operator $H(0) + (gV)_\lambda^\#$ for $\lambda$ in the absolutely continuous
spectrum of $H(0)$ . It is an Anderson-type
random operator at weak disorder, which is known to exhibit localisation in dimension $d = 1$ (see
 Figotin and Pastur \cite[Chapter~15A]{FP}), and is conjectured {(in fact universally accepted by physicists, cf.\ \cite{LTW})} to exhibit localisation in dimension $d=2$, and
delocalisation in dimension $d \geq 3$. Thus the same properties should hold for the
trimmed random
Schr\"odinger operator $H^\Sha + \mathbb{1}$.

Finally observe that the above reasoning is not limited to the hedgehog lattice, and can be extended to more realistic
lattices (such as $\mathbb{Z}^d$). Indeed, the Schur--Banachiewicz formula can still be applied, relating the resolvent
of $H(g)$, $g \gg 1$, to the resolvent of a more complicated Anderson-type operator at weak coupling, which
should share the phenomenological properties of the usual Anderson model.

\paragraph{Acknowledgment}
We are grateful to John Imbrie for helpful discussions  which led to the current
version of Lemma~\ref{l:bk}, and to Lana Jitomirskaya
for comments on the content of Section~\ref{s:anom} {\footnotesize(OT HEE 3HAEM: CBEPXPOCT OTBETOB K (1.5) HEECTECTBEH. HO HE 3HAKOMO: BOBEK HE TECEH 3A3OP B TEOPEME~3?)}.

\end{document}